%% file: Capacity_Bounds_WDM.tex
\documentclass[a4paper, oneside, twocolumn, notitlepage, 10pt]{IEEEtran}
\usepackage{amsthm,amsmath}
\usepackage{booktabs}
\usepackage{pgfplots}
\usepackage{epstopdf}
\usepackage{graphicx}
\usepackage{amsfonts}
\usepackage{amssymb}
\usepackage{float}
\usepackage{tikz}
\usetikzlibrary{shapes,arrows,backgrounds}
\usetikzlibrary{calc,fit,arrows}
\usepackage{subfig}
\usepackage{caption}
\usepackage{adjustbox}
\usepackage{lineno}
\usepackage{setspace}
\usepackage{stfloats} 
\usepackage{array, multirow}
\newcolumntype{C}[1]{>{\centering\arraybackslash}p{#1}}
\usetikzlibrary{shapes}
\usepgfplotslibrary{groupplots}
\pgfplotsset{compat=1.13}
\usepackage{setspace}
\newtheorem{theorem}{Theorem}

\allowdisplaybreaks

\begin{document}


\title{A Capacity Region Outer Bound for the Two-User Perturbative Nonlinear Fiber Optical Channel}%

\author{
\IEEEauthorblockN{Viswanathan~Ramachandran, Astrid~Barreiro, Gabriele~Liga, and Alex~Alvarado}\\
\IEEEauthorblockA{Eindhoven University of Technology, 5600 MB Eindhoven, The Netherlands}
}
\maketitle

\begin{abstract}
    We study a nonlinear fiber optical channel impaired by cross-phase modulation and dispersion from the viewpoint of an interference channel. We characterize an outer bound on the capacity region of simultaneously achievable rate pairs, assuming a two-user perturbative channel model. 
\end{abstract}


\section{Introduction}
In wavelength division multiplexing (WDM), independent data are multiplexed into a single optical fiber using several transmitters, with corresponding demultiplexing at the receiver side. The nonlinear {Kerr effect} in an optical fiber causes the signal in one wavelength to interfere with the signals in other wavelengths. A very difficult form of interchannel nonlinearity to compensate for is cross-phase modulation (XPM), which combined with chromatic dispersion introduces memory into the channel. The fundamental limits of communication over such channels with nonlinearity and dispersion fall within the domain of \textit{multiuser information theory}, which investigates tradeoffs between the rates at which \textit{all different users} can operate.

The dominant paradigm in optical multi-user channels, however, has been the study of achievable rates for individual users in the system. For instance, Secondini and Forestieri~\cite{secondini2016scope} studied the capacity of a single user in a WDM system and proved that it grows unbounded with power as opposed to Gaussian achievable information rates. Agrell and Karlsson~\cite{agrell2015influence} studied the impact of different {behavioral models} for the interfering users on the capacity of a specific user in the system. Such works attempt to reduce the analysis of a multi-user problem to more familiar single-user problems. Moreover, as far as capacity upper bounds are concerned, the only known ones for a general nonlinear Schr{\"o}dinger (NLSE) channel are that of Yousefi \emph{et al.}~\cite{yousefi2015upper}, Kramer \emph{et al.}~\cite{kramer2015upper} and Keykhosravi \emph{et al.}~\cite{keykhosravi2017tighter}, the first two of which coincide with AWGN channel capacity.

In the multi-user information theory literature, multiple one-to-one communications over a shared medium with crosstalk between the users is known as an \textit{interference channel}~\cite[Chapter 6]{el2011network}. Optical interference channels have attracted some attention in the literature. For instance, Taghavi \emph{et al.}~\cite{taghavi2006multiuser} analyzed the benefits of multi-user detection in WDM systems by modelling it as a \emph{multiple access channel} (which is nothing but full receiver cooperation in an interference channel). Ghozlan and Kramer~\cite{ghozlan2017models} studied an interference channel model based on logarithmic perturbation and introduced \emph{interference focusing} to achieve the optimal pre-log factors.

We observe that a study of the set of \emph{simultaneously achievable} rates that captures the contention amongst the different users accessing the optical channel transmission resources has received very little attention in the literature, with the aforementioned exceptions~\cite{taghavi2006multiuser} ~\cite{ghozlan2017models}. Moreover, capacity upper bounds are as of today missing in the framework of optical multi-user channels.

In this paper, we propose a novel outer bound on the \emph{capacity region}, i.e., the region of all simultaneously achievable information rates, of a multi-user/WDM channel where both transmitters and receivers are independently operated. The investigated channel is a simplified version (memoryless) of the perturbative multi-user model where both chromatic dispersion and Kerr nonlinearity are assumed during fiber propagation.

\section{Channel Model}
We study the two-user WDM system shown in Fig.~\ref{fig:ICmodel}, where the interference channel $p(y,z|x,w)$ encompasses the electro-optical (E-O) conversion, the WDM multiplexing, the physical channel, the WDM demultiplexing, and optical-electrical (E-O) conversion. We assume single-polarization transmission over a single span of standard single mode fiber (SSMF). The output at the receiver of user$-1$ can be expressed using a first-order regular perturbative discrete-time model~\cite{mecozzi2012nonlinear,dar2013properties}
\begin{align}
Y_k=X_k+\!\!\sum_{l,m,p}\!\!c_{l,m,p}^{(x)}W_{k-m}W_{k-p}^{*}X_{k-l}+\!N_{k}^{(x)},\! \label{eq:model}
\end{align}
where $X_k$ represents the input of user$-1$ at time instant $k \in [1:n]$, $W_k$ is the input of the other user at instant $k$, while $W_{k-m}$ represents the corresponding input at a time lag of $m$. Note that an $n-$ length block of output symbols is denoted by $Y^n \triangleq (Y_1,Y_2,\cdots,Y_n)$ in Fig.~\ref{fig:ICmodel} and henceforth in the sequel. The complex channel coefficients $c_{l,m,p}^{(x)}$ are given in~\cite{dar2013properties}, are computed numerically, and depend on the properties of the optical link and the transmission parameters. In \eqref{eq:model}, $N_{k}^{(x)}$ models amplified spontaneous emission (ASE) noise from the erbium-doped amplifier (EDFA). The ASE noise is assumed to be circularly symmetric complex Gaussian with mean zero and variance $\sigma^2$ per complex dimension. Similarly, the second channel is specified by
\begin{align}
Z_k=W_k+\!\!\sum_{l,m,p}\!\!c_{l,m,p}^{(w)}X_{k-m}X_{k-p}^{*}W_{k-l}+\!N_{k}^{(w)}\!.\! \label{eq:model2}
\end{align}
We assume length-$n$ codewords with per-codeword power constraints:
\begin{align}
\frac{1}{n} \sum_{k=1}^n \mathbb{E}[|X_k|^2] &\leq P_1, \frac{1}{n} \sum_{k=1}^n \mathbb{E}[|W_k|^2] &\leq P_2.
\end{align}
Though we assume that self-phase modulation (SPM) is ideally compensated in \eqref{eq:model}--\eqref{eq:model2}, the bounds can be generalized to take into account SPM as well as XPM.

\begin{figure*}
\begin{center}
\scalebox{0.8}{\input{Figure_channel_model_ECOC}}
\caption{Interference channel model for WDM transmission.}
\label{fig:ICmodel}
\end{center}
\end{figure*}
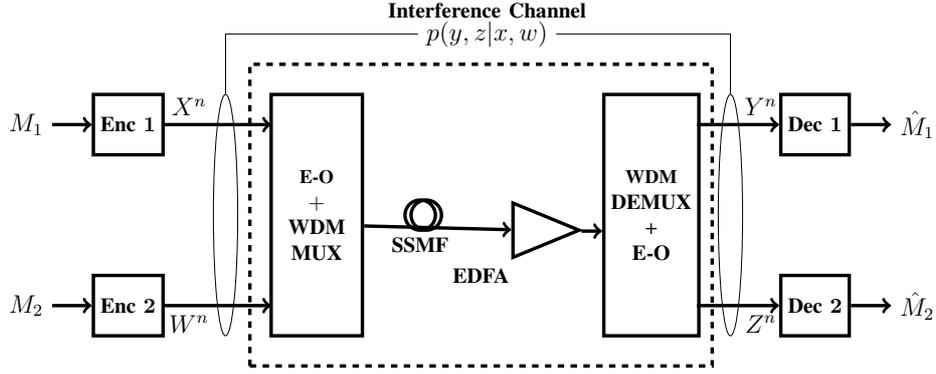
\section{Capacity Region Outer Bounds}
In this section, we analyze the region of simultaneously achievable rate pairs $(R_1,R_2)$. Before stating the outer bound, we need some definitions. An $(n,2^{nR_1},2^{nR_2})$ code for this channel consists of two message sets $[1:2^{nR_1}]$ and $[1:2^{nR_2}]$, two encoders where Enc $1$ maps a message $M_1 \in [1:2^{nR_1}]$ into a codeword $X^n(M_1)$ and Enc $2$ maps a message $M_2 \in [1:2^{nR_2}]$ into a codeword $W^n(M_2)$, and two decoders where Dec-$1$ assigns an estimate $\hat{M}_1$ (or an error message) to each received sequence $Y^n$ and Dec-$2$ assigns an estimate $\hat{M}_2$ (or an error message) to each received sequence $Z^n$. The probability of error is defined as
\begin{align}
P_e^{(n)} = \textup{Pr}((\hat{M}_1(Y^n),\hat{M}_2(Z^n)) \neq (M_1,M_2)).
\end{align}
A rate pair $(R_1,R_2)$ is said to be \emph{achievable} if there exists a sequence of $(n,2^{nR_1},2^{nR_2})$ codes such that $\lim_{n \to \infty} P_e^{(n)} = 0$. The capacity region $\mathcal{C}$ is the closure of the set of achievable rate pairs $(R_1,R_2)$. 

Capacity region analysis for the full model specified by \eqref{eq:model}--\eqref{eq:model2} is involved due to the channel memory. Hence as a  first step towards capacity region outer bounds, we focus on the following approximation
\begin{align}
\tilde{Y}_k&=X_k+c_{0,0,0}^{(x)}|W_k|^2 X_k+N_k^{(x)}, \label{eq:memless1}\\
\tilde{Z}_k&=W_k+c_{0,0,0}^{(w)}|X_k|^2 W_k+N_k^{(w)}. \label{eq:memless2}
\end{align} 
In particular, note that the approximated model ignores all the
elements in \eqref{eq:model}--\eqref{eq:model2} that introduce memory.

\begin{table}[t]
\def\arraystretch{1.0}%
\caption{Model Parameters}
\centering
\footnotesize
\begin{tabular}{c c}
\toprule
Parameter & Value \\ 
\midrule
Memory length & $5$ \\
Distance & $250 \: \textrm{km}$ \\
Nonlinearity parameter $\gamma$ & $1.2 \: \textrm{W}^{-1} \textrm{km}^{-1}$ \\
Baud Rate & $32 \: \textrm{Gbaud}$ \\
Fiber attenuation $\alpha$ & $0.2 \: \textrm{dB/km}$ \\
Group velocity dispersion $\beta_2$ & $-21.7 \: \textrm{ps\textsuperscript{2}/km}$ \\
\bottomrule
\end{tabular}
\label{table:param} 
\end{table}
\normalsize

We now present an outer bound on the capacity region of the approximated channel in \eqref{eq:memless1}--\eqref{eq:memless2}. For convenience, let us denote \smash{$c_{0,0,0}^{(x)}\triangleq g_{(x)}=g_{(x)}^{R}+jg_{(x)}^{I}$} and \smash{$c_{0,0,0}^{(w)}\triangleq g_{(w)}=g_{(w)}^{R}+jg_{(w)}^{I}$} in terms of the respective real and imaginary parts.

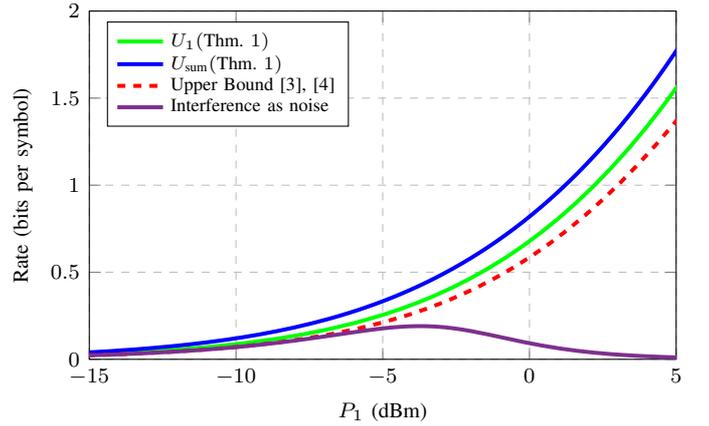
\begin{figure}[!t]
\vspace{-6.5cm}
\begin{center}
\input{fig2}
\end{center}
\caption{Outer bounds versus input power}
\label{fig:OBsnr}
\end{figure}

\begin{figure*}[!t]
\input{fig3.tikz}
\caption{Outer bound rate regions for a transmitted power per user of (a) $-2.9$ dBm , (b) $0$ dBm, and (c) $2$ dBm.}
\label{fig:OBregion}
\end{figure*}
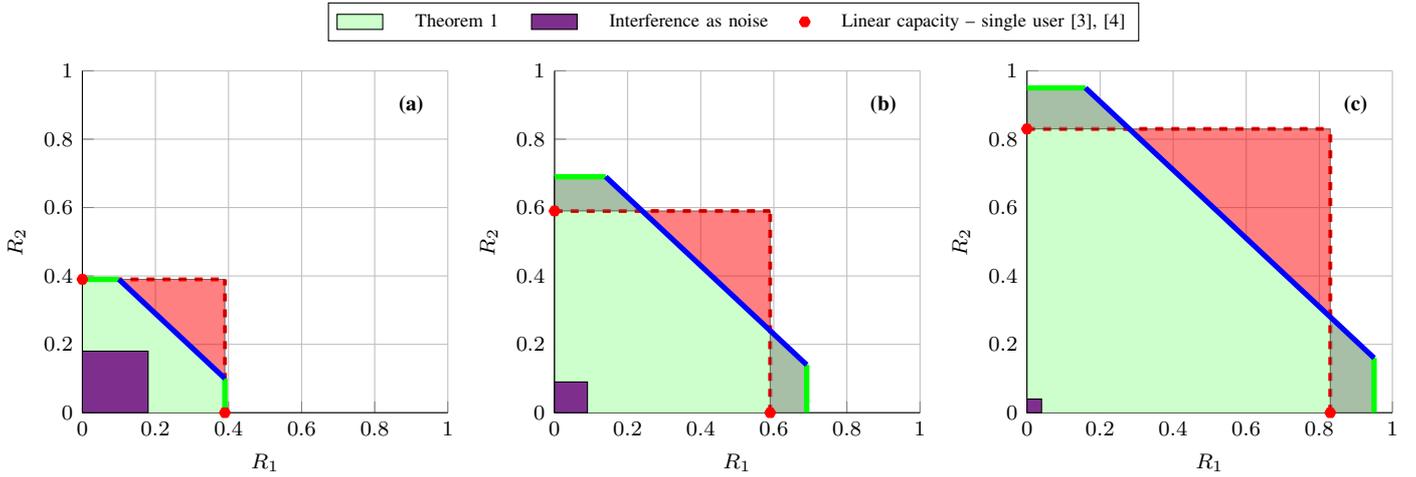

\begin{theorem} \label{thm:OBwdm}
The capacity region of the approximated channel \eqref{eq:memless1}--\eqref{eq:memless2} is upper bounded by the set of $(R_1,R_2)$ pairs such that
\begin{align}
&R_1 \leq U_1 \triangleq \log\left(1+\frac{(1+2g_{(x)}^{R}P_2+2|g_{(x)}|^2 P_2^2)P_1}{2 \sigma^2}\right),\\
&R_2 \leq U_2 \triangleq \log\left(1+\frac{(1+2g_{(w)}^{R}P_1+2|g_{(w)}|^2 P_1^2)P_2}{2 \sigma^2}\right), \\
&R_1+R_2 \leq U_{\textrm{sum}} \notag\\
&\triangleq 2\log\left(\frac{2^{U_1}+2^{U_2}}{2}+\frac{|g_{(x)}|^2 P_2^2 P_1+|g_{(w)}|^2 P_1^2 P_2}{2\sigma^2}\right).
\end{align}
\end{theorem}
\begin{proof}
See Appendix~\ref{proof.Theorem}.
\end{proof}

\section{Numerical Results}
The parameters used in our numerical results are summarized in Table \ref{table:param}. The rate bounds on user$-1$ ($U_1$) and the sum rate bound $U_{\textrm{sum}}$ in Theorem~\ref{thm:OBwdm} are plotted in Fig. \ref{fig:OBsnr} against the input power for the symmetric case of $P_1=P_2$. For comparison, the linear Gaussian capacity upper bound given by $\log(1+P_1/(2\sigma^2))$ from Yousefi \emph{et al.}~\cite{yousefi2015upper,kramer2015upper} is also shown, along with the lower bound obtained by treating the interference terms in \eqref{eq:model}--\eqref{eq:model2} as Gaussian noise.

In Figs. \ref{fig:OBregion}(a)--(c), we plot the trade-off between the rates of the two users for fixed (and equal) powers of $-2.9$ dBm, $0$ dBm and $2$ dBm, respectively. In Fig. \ref{fig:OBregion}(a) for instance, it can be seen that when user-$1$ transmits at its maximum rate of $0.39$ bits, user-$2$ must lower its rate to a maximum of $0.1$ bits to respect the sum rate constraint. The exhibited rate trade-off region must be contrasted with the more equitable strategies normally employed in optical systems that operate at a rate of $\frac{0.1+0.39}{2}$ bits for each user, i.e., at the mid-point of the dominant face of the pentagon. We have also depicted the respective rate regions obtained by treating the interference terms in \eqref{eq:model}--\eqref{eq:model2} as Gaussian noise for comparison (in blue), that only achieve the rectangular region defined by the single-user rate constraints. It is observed that with increasing powers (the evolution from Fig. \ref{fig:OBregion}(a) to Fig. \ref{fig:OBregion}(c)), the outer bound regions grow in size while the \emph{interference as noise} regions eventually vanish in the highly nonlinear regime.

The rectangular region implied by the $\log(1+P_1/(2\sigma^2))$ outer bound of Yousefi \emph{et al.}~\cite{yousefi2015upper,kramer2015upper} are shown by red dotted lines in Figs. \ref{fig:OBregion}(a)--(c), while the triangular regions shaded in red represent the additional points that are ruled out of the capacity region by our outer bound in Theorem.~\ref{thm:OBwdm}. On the other hand, the trapezoidal portions shaded in gray in Figs. \ref{fig:OBregion}(b) and \ref{fig:OBregion}(c) represent the set of points already ruled out by the outer bound~\cite{yousefi2015upper,kramer2015upper}.

\section{Conclusions}

We took a multi-user information theoretic view of a two-channel WDM system impaired by XPM and dispersion, and derived a capacity region outer bound. This is the very first step towards characterizing the tension between the rates of the interfering users. Future works include extension of the results to more than two users and obtaining tight achievable regions or inner bounds.


\section*{Acknowledgements}
The work of V.  Ramachandran, A. Barreiro and A. Alvarado has received funding from the European Research Council (ERC) under the European Union's Horizon 2020 research and innovation programme (grant agreement No 757791). The work of G.~Liga is funded by the EuroTechPostdoc programme under the European Union's Horizon 2020 research and innovation programme (Marie Sk\l{}odowska-Curie grant agreement No 754462).


\appendices

\section{Proof of Theorem \ref{thm:OBwdm}}\label{proof.Theorem}
We establish the outer bound using information theoretic inequalities. The rate of user$-1$ can be upper bounded as follows.
\begin{align}
&nR_1 = H(M_1) \stackrel{(a)}= H(M_1|W^n) \notag\\
&= H(M_1|W^n)-H(M_1|W^n,\tilde{Y}^n)+H(M_1|W^n,\tilde{Y}^n) \notag\\
&\stackrel{(b)}\leq I(M_1;\tilde{Y}^n|W^n)+1+P_{e}^{(n)} nR_1  \notag\\
&= I(M_1;\tilde{Y}^n|W^n)+n(1/n+P_{e}^{(n)} R_1) \notag\\
&\stackrel{(c)}= I(M_1;\tilde{Y}^n|W^n)+n\epsilon_n \notag\\
&\stackrel{(d)}\leq I(X^n;\tilde{Y}^n|W^n)\!+\!n\epsilon_n  = h(\tilde{Y}^n|W^n)-\!h({N^{(x)}}^n)+\!n\epsilon_n \notag\\
&\stackrel{(e)}\leq \sum_{i=1}^n h(\tilde{Y}_i|W_i)-\sum_{i=1}^n h(N_i^{(x)})+n\epsilon_n \notag\\
&= \sum_{i=1}^n \mathbb{E}[h(\tilde{Y}_i|W_i=w)]-\sum_{i=1}^n h(N_i^{(x)})+n\epsilon_n \notag\\
&\stackrel{(f)}\leq \sum_{i=1}^n \mathbb{E}\left[\frac{1}{2}\log\left(\frac{ \textup{det}(\textup{cov}(\tilde{Y}_i^R,\tilde{Y}_i^I|W_i=w))}{\textup{det}(\textup{cov}({N_i^{(x)}}^R,{N_i^{(x)}}^I))}\right)\right]+n\epsilon_n\!, \label{eq:conv1}
\end{align}
where (a) follows since $M_1$ is independent of $W^n$, (b) follows from Fano's inequality, (c) follows by defining $\epsilon_n =(1/n+P_{e}^{(n)} R_1)$ with $\epsilon_n \xrightarrow{n \to \infty} 0$, (d) follows from the data processing inequality since $M_1 \to X^n \to \tilde{Y}^n$ forms a Markov chain, (e) follows since conditioning does not increase the entropy and the fact that the additive noise is i.i.d., while (f) follows from the fact that Gaussian random vectors maximize the differential entropy under a covariance constraint.
Similarly, we obtain for the second user
\begin{align}
nR_2 &\leq \sum_{i=1}^n \mathbb{E}\left[\frac{1}{2}\log\left(\frac{ \textup{det}(\textup{cov}(\tilde{Z}_i^R,\tilde{Z}_i^I|X_i=x))}{\textup{det}(\textup{cov}({N_i^{(w)}}^R,{N_i^{(w)}}^I))}\right)\right]+n\epsilon_n. \label{eq:conv2}
\end{align}
For the sum rate upper bound, we can write the following chain of inequalities following similar reasoning.
\begin{align}
&n(R_1+R_2) = H(M_1,M_2)\notag\\
&= H(M_1,M_2)-H(M_1,M_2|\tilde{Y}^n,\tilde{Z}^n)+H(M_1,M_2|\tilde{Y}^n,\tilde{Z}^n) \notag\\
&\leq I(M_1,M_2;\tilde{Y}^n,\tilde{Z}^n)+n\epsilon_n  \notag\\
&\leq I(X^n,W^n;\tilde{Y}^n,\tilde{Z}^n)+n\epsilon_n  \notag\\
&= h(\tilde{Y}^n,\tilde{Z}^n)-h({N^{(x)}}^n,{N^{(w)}}^n)+n\epsilon_n \notag\\
&\leq \sum_{i=1}^n h(\tilde{Y}_i,\tilde{Z}_i)-\sum_{i=1}^n h(N_i^{(x)},N_i^{(w)})+n\epsilon_n \notag\\
&\leq \sum_{i=1}^n \frac{1}{2}\log\left(\frac{\textup{det}(\textup{cov}(\tilde{Y}_i^R,\tilde{Y}_i^I,\tilde{Z}_i^R,\tilde{Z}_i^I))}{\textup{det}(\textup{cov}({N_i^{(x)}}^R,{N_i^{(x)}}^I,{N_i^{(w)}}^R,{N_i^{(w)}}^I))}\right) \notag\\
&\phantom{wwwwwww}+n\epsilon_n. \label{eq:conv3}
\end{align}

It now remains to bound the $\log(\textup{det}(\cdot))$ terms in expressions \eqref{eq:conv1}--\eqref{eq:conv3}. Equations \eqref{eq:memless1} and \eqref{eq:memless2} can be expressed in terms of their respective real and imaginary components. For this step, let us denote
\begin{align*}
&X_i=X_i^R+jX_i^I, \:\: W_i=W_i^R+jW_i^I, \\ &\tilde{Y}_i=\tilde{Y}_i^R+j\tilde{Y}_i^I, \:\:\tilde{Z}_i=\tilde{Z}_i^R+j\tilde{Z}_i^I,\\ 
&N_i^{(x)}={N_i^{(x)}}^{R}+j{N_i^{(x)}}^{I},\\ &N_i^{(w)}={N_i^{(w)}}^{R}+j{N_i^{(w)}}^{I}.
\end{align*}
Then it follows that
\begin{align}
\tilde{Y}_i^R &= (1+|W_i|^2 g_{(x)}^{R}) X_i^R-|W_i|^2 g_{(x)}^{I} X_i^I+{N_i^{(x)}}^{R}, \\
\tilde{Y}_i^I &= (1+|W_i|^2 g_{(x)}^{R}) X_i^I+|W_i|^2 g_{(x)}^{I} X_i^R+{N_i^{(x)}}^{I}, \\
\tilde{Z}_i^R &= (1+|X_i|^2 g_{(w)}^{R}) W_i^R-|X_i|^2 g_{(w)}^{I} W_i^I+{N_i^{(w)}}^{R}, \\
\tilde{Z}_i^I &= (1+|X_i|^2 g_{(w)}^{R}) W_i^I+|X_i|^2 g_{(w)}^{I} W_i^R+{N_i^{(w)}}^{I}.
\end{align}
Let $\mathbb{E}[(X_i^R)^2]=p_{1i}^R$ and $\mathbb{E}[(X_i^I)^2]=p_{1i}^I$ such that $p_{1i}^R+p_{1i}^I \leq P_{1i}$. Similarly, denote $\mathbb{E}[(W_i^R)^2]=p_{2i}^R$ and $\mathbb{E}[(W_i^I)^2]=p_{2i}^I$ such that $p_{2i}^R+p_{2i}^I \leq P_{2i}$. Hence we can write
\begin{align}
&\textup{det}(\textup{cov}(\tilde{Y}_i^R,\tilde{Y}_i^I|W_i=w)) \notag\\
&=\textup{det}\left(\textup{cov}\left((1+|w|^2 g_{(x)}^{R}) X_i^R-|w|^2 g_{(x)}^{I} X_i^I+{N_i^{(x)}}^{R}, \right.\right. \notag\\
&\phantom{www}\left.\left.(1+|w|^2 g_{(x)}^{R}) X_i^I+|w|^2 g_{(x)}^{I} X_i^R+{N_i^{(x)}}^{I}|W_i=w\right)\right) \notag\\
&=\textup{det}\left(\textup{cov}\left((1+|w|^2 g_{(x)}^{R}) X_i^R-|w|^2 g_{(x)}^{I} X_i^I+{N_i^{(x)}}^{R}, \right.\right. \notag\\
&\phantom{www}\left.\left.(1+|w|^2 g_{(x)}^{R}) X_i^I+|w|^2 g_{(x)}^{I} X_i^R+{N_i^{(x)}}^{I}\right)\right) \notag\\
&\stackrel{(a)}\leq \!\frac{1}{4}\!\left(\!\begin{aligned}&\textup{var}((1+|w|^2 g_{(x)}^{R}) X_i^R-|w|^2 g_{(x)}^{I} X_i^I+{N_i^{(x)}}^{R})\\ &+\!\textup{var}(\!(1+|w|^2 g_{(x)}^{R}) X_i^I+|w|^2 g_{(x)}^{I} X_i^R+{N_i^{(x)}}^{I}\!)\end{aligned}\!\!\right)^2 \notag\\
&= \left(\frac{(1+2g_{(x)}^{R}|w|^2+|g_{(x)}|^2|w|^4)(p_{1i}^R+p_{1i}^I)+2\sigma^2}{2}\right)^2 \notag\\
&\stackrel{(b)}\leq \left(\frac{(1+2g_{(x)}^{R}|w|^2+|g_{(x)}|^2|w|^4)P_{1i}}{2}+\sigma^2\right)^2, \label{eq:conv4}
\end{align}
where (a) follows since $\textup{det}(A) \leq \left(\frac{\textup{trace}(A)}{n}\right)^n$ for any $n\times n$ square matrix $A$, while (b) follows from the power constraint. From \eqref{eq:conv1} and \eqref{eq:conv4},
\begin{align}
&n(R_1-\epsilon_n) \notag\\
&\leq \sum_{i=1}^n \mathbb{E}\left[\log\left(1\!+\!\frac{(1+2g_{(x)}^{R}|w|^2+|g_{(x)}|^2|w|^4)P_{1i}}{2\sigma^2}\right)\right] \notag\\
&\stackrel{(a)}\leq n\mathbb{E}\left[\log\left(1\!+\!\frac{(1+2g_{(x)}^{R}|w|^2+|g_{(x)}|^2|w|^4)P_{1}}{2\sigma^2}\right)\right] \notag\\
&\stackrel{(b)}\leq n\log\left(1\!+\!\frac{(1+2g_{(x)}^{R}\mathbb{E}[|w|^2]+|g_{(x)}|^2\mathbb{E}[|w|^4])P_{1}}{2\sigma^2}\right) \notag\\
&= n\log\left(1+\frac{(1+2g_{(x)}^{R} P_2+2|g_{(x)}|^2 P_2^2)P_{1}}{2\sigma^2}\right) = nU_1,
\end{align}
where (a) and (b) follow from the concavity of the $\log(\cdot)$ and Jensen's inequality. Dividing throughout by $n$ and letting $n \to \infty$ (which makes $\epsilon_n \to 0$) completes the proof for the bound on the individual user rate. Similarly, we obtain
\begin{align}
&\textup{det}(\textup{cov}(\tilde{Z}_i^R,\tilde{Z}_i^I|X_i=x)) \notag\\
&\leq \left(\frac{(1+2g_{(w)}^{R}|x|^2+|g_{(w)}|^2|x|^4)P_{2i}}{2}+\sigma^2\right)^2, \label{eq:conv5}
\end{align}
which along with \eqref{eq:conv2} gives
\begin{align}
n(R_2-\epsilon_n) &\leq n\log\left(1+\frac{(1+2g_{(w)}^{R} P_1+2|g_{(w)}|^2 P_1^2)P_{2}}{2\sigma^2}\right) \notag\\
&= nU_2.
\end{align}
Finally for the sum rate term, the determinant of the covariance matrix can be computed along similar lines
\begin{align}
&\textup{det}(\textup{cov}(\tilde{Y}_i^R,\tilde{Y}_i^I,\tilde{Z}_i^R,\tilde{Z}_i^I)) \notag\\
&\leq \left(\frac{\textup{var}(\tilde{Y}_i^R)+\textup{var}(\tilde{Y}_i^I)+\textup{var}(\tilde{Z}_i^R)+\textup{var}(\tilde{Z}_i^I)}{4}\right)^4 \notag\\
&= \left(\sigma^2+\frac{1}{4}\left(\!\begin{aligned}&(1+2g_{x}^{R} P_{2i}+4|g_{x}|^2 P_{2i}^2)P_{1i}\\ &+(1+2g_{w}^{R} P_{1i}+4|g_{w}|^2 P_{1i}^2)P_{2i}\end{aligned}\right)\right)^4.  \label{eq:conv6}
\end{align}
Expressions \eqref{eq:conv3} and \eqref{eq:conv6} together give
\begin{align}
&n(R_1+R_2-\epsilon_n) \notag\\
&\leq \sum_{i=1}^n 2\log\!\left(\!1\!+\!\frac{1}{4\sigma^2}\!\!\left(\!\begin{aligned}&(1+2g_{x}^{R} P_{2i}+4|g_{x}|^2 P_{2i}^2)P_{1i}\\ &+(1+2g_{w}^{R} P_{1i}+4|g_{w}|^2 P_{1i}^2)P_{2i}\end{aligned}\!\right)\!\!\right) \notag\\
&\leq 2n\log\!\left(\!1\!+\!\frac{1}{4\sigma^2}\!\!\left(\!\begin{aligned}&(1+2g_{x}^{R} P_{2}+4|g_{x}|^2 P_{2}^2)P_{1}\\ &+(1+2g_{w}^{R} P_{1}+4|g_{w}|^2 P_{1}^2)P_{2}\end{aligned}\!\right)\!\!\right) \notag\\
&= 2n\log\left(\frac{2^{U_1}+2^{U_2}}{2}+\frac{|g_{x}|^2 P_2^2 P_1+|g_{w}|^2 P_1^2 P_2}{2\sigma^2}\right).
\end{align}
Dividing throughout by $n$ and letting $n \to \infty$ (which makes $\epsilon_n \to 0$) completes the proof.

\nocite{}
\bibliographystyle{IEEEtran}
\bibliography{mylitoptics.bib}

\end{document}

%% file: Figure_channel_model_ECOC.tex
\begin{tikzpicture}[ultra thick]
\node (e1) at (-2.4,1.5) [rectangle, draw, minimum height=1.0cm]{\textbf{Enc}~$\boldsymbol{1}$}; 
\node (e2) at (-2.4,-1.5) [rectangle, draw, minimum height=1.0cm]{\textbf{Enc}~$\boldsymbol{2}$}; 
\node (e3) at (0.7,0) [rectangle, draw, minimum height=4cm, minimum width=1.5cm, align=center]{\small \textbf{E-O} \\\textbf{$+$} \\ \textbf{WDM}\\ \textbf{MUX}}; 
\node (e4) at (6.2,0) [rectangle, draw, minimum height=4cm, minimum width=1.5cm,align=center]{\small \textbf{WDM}\\ \textbf{DEMUX}\\+\\\textbf{E-O}};
\node (c1) at (2.5,0) [circle, draw, minimum height=0.5cm]{};
\node (c2) at (2.4,0) [circle, draw, minimum height=0.5cm]{};
\node[isosceles triangle, draw, minimum size =.9cm] (T)at (4.2,-0.25){};
\node (d1) at (8.3,1.5) [rectangle, draw, right, minimum height=1.0cm]{\textbf{Dec}~$\boldsymbol{1}$};
\node (d2) at (8.3,-1.5) [rectangle, draw, right, minimum height=1.0cm]{\textbf{Dec}~$\boldsymbol{2}$};
\node () at (2.4,-0.5) {\textbf{SSMF}};
\node () at (3.4,-1) {\textbf{EDFA}};
\node (e5) at (3.4,0) [rectangle, dashed, draw, minimum height=5cm, minimum width=7.6cm]{};
\node () at (-1.4,1.8) {\large $X^n$};
\node () at (-1.4,-1.8) {\large $W^n$};
\node () at (8,1.8) {\large $Y^n$};
\node () at (8,-1.8) {\large $Z^n$};
\node () at (3.5,3.4) {\textbf{Interference Channel}};
\node (distr) at (3.5,3) {\large $p(y,z|x,w)$};
\node[draw, ellipse, minimum width=.4cm, minimum height=4cm, thin] at (-.8,0) (ell1) {};
\node[draw, ellipse, minimum width=.4cm, minimum height=4cm, thin] at (7.5,0) (ell2) {};

\draw[-,thin] (distr)--(distr-|ell1)--(ell1);
\draw[-,thin] (distr)--(distr-|ell2)--(ell2);

\draw[<-] (e1) --++(-1.25,0) node[left]{\large $M_1$};
\draw[<-] (e2) --++(-1.25,0) node[left]{\large $M_2$};
\draw[->] (d1) --++(1.25,0) node[right]{\large $\hat{M}_1$};
\draw[->] (d2) --++(1.25,0) node[right]{\large $\hat{M}_2$};
\draw[->] (e3.-13)  -- (T);
\draw[->] (T.-1)  -- (T.-1-|e4.west);
\draw[->] (e1) --++(2.35,0) node[above]{};
\draw[->] (e2) --++(2.35,0) node[above]{};
\draw[<-] (d1.west) --++(-1.35,0) node[midway, above]{};
\draw[<-] (d2.west) --++(-1.35,0) node[midway, above]{};
\end{tikzpicture}

%% file: fig2.tex
%
%
\definecolor{mycolor1}{rgb}{0.00000,0.44700,0.74100}%
\definecolor{mycolor2}{rgb}{0.85000,0.32500,0.09800}%
\definecolor{mycolor3}{rgb}{0.92900,0.69400,0.12500}%
\definecolor{mycolor4}{rgb}{0.49400,0.18400,0.55600}%
\begin{tikzpicture}[tight background]

\begin{axis}[%
every axis/.append style={font=\footnotesize},
width=1.05\columnwidth,
height=0.7\columnwidth,
xmin=-15,
xmax=5,
xlabel style={font=\color{white!15!black}},
xlabel={$P_1~\text{(dBm)}$},
ymin=0,
ymax=2.0,
ylabel style={font=\color{white!15!black}},
ylabel={Rate (bits per symbol)},
ylabel near ticks,
xlabel near ticks,
axis background/.style={fill=white},
xtick={-20,-15,...,10},
xmajorgrids,
ymajorgrids,
legend style={
legend pos=north west,font=\scriptsize,legend cell align=left,row sep=-0.5ex,grid style={dashed}
}
]
\addplot [color=green, line width=1.5pt]
  table[row sep=crcr]{%
-20	0.0071973967862462\\
-16.1	0.02184965548647\\
-14	0.0353364412960744\\
-12.5	0.0495938745649482\\
-11	0.0694290000759565\\
-9.8	0.090658280559893\\
-8.6	0.118062305309039\\
-7.7	0.143618547223284\\
-6.8	0.174328456050464\\
-5.9	0.211077093729315\\
-5	0.25484530148557\\
-4.4	0.288448564591842\\
-3.8	0.325981371251466\\
-3.2	0.367792584349452\\
-2.6	0.414241135875667\\
-2	0.465693032006364\\
-1.4	0.522518163803468\\
-0.800000000000001	0.585087085220593\\
-0.199999999999999	0.653767952379344\\
0.399999999999999	0.728923843082679\\
1	0.810910689918064\\
1.6	0.900076060663551\\
2.2	0.996759002677486\\
2.8	1.10129113035406\\
3.4	1.21399907384693\\
4	1.33520832130178\\
4.6	1.46524837565144\\
5.2	1.60445901319503\\
5.8	1.75319728132463\\
6.4	1.91184471877165\\
7	2.08081414129086\\
7.6	2.26055523115548\\
7.9	2.35461669423647\\
8.2	2.45155812393972\\
8.5	2.55144724298214\\
8.8	2.65435422045029\\
9.1	2.7603515016426\\
9.4	2.86951357470533\\
9.7	2.98191667325082\\
10	3.09763841536966\\
10.3	3.21675738069183\\
10.6	3.33935262838579\\
10.9	3.46550316018452\\
11.2	3.59528733367396\\
11.5	3.72878223215643\\
11.8	3.86606299839892\\
12.1	4.00720214048406\\
12.4	4.1522688187908\\
12.7	4.30132812383161\\
13	4.45444035524917\\
13.3	4.61166031270873\\
13.6	4.77303660969138\\
13.9	4.93861102127169\\
14.2	5.10841787682604\\
14.5	5.28248350824114\\
14.8	5.46082576355913\\
15.1	5.64345359509818\\
15.4	5.83036672993068\\
15.7	6.02155542920394\\
16	6.21700034118537\\
16.3	6.41667245115359\\
16.6	6.62053312939789\\
16.9	6.82853427669803\\
17.2	7.04061856480258\\
};
\addlegendentry{$U_1 (\textrm{Thm.}~\ref{thm:OBwdm})$}

\addplot [color=blue, line width=1.5pt]
  table[row sep=crcr]{%
-15.2	0.0379125560546889\\
-14.3	0.0467697418672621\\
-13.4	0.0573641977931985\\
-12.8	0.0655835352199325\\
-12.2	0.0748764313222896\\
-11.6	0.0853825133647135\\
-11	0.0972531812111956\\
-10.4	0.110652186383776\\
-9.8	0.125756073158419\\
-9.2	0.142754464471846\\
-8.6	0.16185017984686\\
-8	0.183259178664827\\
-7.4	0.20721035105664\\
-6.8	0.233945444209564\\
-6.5	0.248435980673841\\
-6.2	0.263718850241984\\
-5.9	0.27982772657553\\
-5.6	0.296796991597201\\
-5.3	0.314661712784019\\
-5	0.333457620580884\\
-4.7	0.353221086371279\\
-4.4	0.373989101459673\\
-4.1	0.395799257531689\\
-3.8	0.418689729062431\\
-3.5	0.442699258138667\\
-3.2	0.467867142146165\\
-2.9	0.494233224747175\\
-2.6	0.521837890534004\\
-2.3	0.55072206369068\\
-2	0.580927210924964\\
-1.7	0.612495348845968\\
-1.4	0.645469055857872\\
-1.1	0.679891488517445\\
-0.800000000000001	0.715806402163297\\
-0.500000000000002	0.753258175469517\\
-0.200000000000001	0.792291838408978\\
0.0999999999999996	0.832953102936544\\
0.399999999999999	0.875288395525528\\
0.699999999999999	0.919344890519621\\
0.999999999999998	0.965170543105565\\
1.3	1.01281412057829\\
1.6	1.06232523046976\\
1.9	1.11375434405392\\
2.2	1.16715281373055\\
2.5	1.22257288283549\\
2.8	1.28006768652534\\
3.1	1.3396912425399\\
3.4	1.40149843084933\\
3.7	1.46554496143588\\
4	1.53188732972928\\
4.3	1.60058275949555\\
4.6	1.67168913325515\\
4.9	1.74526491056222\\
5.2	1.82136903469859\\
};
\addlegendentry{$U_{\textrm{sum}} (\textrm{Thm.}~\ref{thm:OBwdm})$}

\addplot [color=red, dashed, line width=1.5pt]
  table[row sep=crcr]{%
-20	0.00719550140420466\\
-16.7	0.0153402894325296\\
-14.3	0.0265547080379775\\
-12.5	0.040004545948463\\
-11	0.0561900791614214\\
-9.8	0.0736234203670634\\
-8.6	0.0962877825090622\\
-7.7	0.117578269562291\\
-6.8	0.143347448357972\\
-5.9	0.174431717104476\\
-5	0.211777127569885\\
-4.1	0.25643411618185\\
-3.5	0.29082754254943\\
-2.9	0.329331090580339\\
-2.3	0.372307841076601\\
-1.7	0.420122849286795\\
-1.1	0.473135989410203\\
-0.5	0.53169382897757\\
0.0999999999999979	0.596120799682655\\
0.699999999999999	0.666710040903645\\
1.3	0.743714386376759\\
1.9	0.827338026689379\\
2.5	0.917729395366322\\
3.1	1.01497578431492\\
3.7	1.11910009348149\\
4.3	1.23005996789747\\
4.9	1.34774939041757\\
5.5	1.47200260443837\\
6.1	1.6026000636319\\
6.7	1.73927596778322\\
7.3	1.88172686012926\\
7.9	2.02962073727685\\
8.5	2.18260615331427\\
9.1	2.34032087322535\\
9.7	2.50239973127271\\
10.3	2.66848146124816\\
11.2	2.92434386727662\\
12.1	3.18730067758539\\
13	3.45632143459955\\
13.9	3.7304817755326\\
15.1	4.10263149806735\\
16.3	4.48084518214059\\
17.8	4.96014002109612\\
19.6	5.54227482851593\\
20.2	5.7375901658105\\
};
\addlegendentry{Upper Bound~\cite{yousefi2015upper,kramer2015upper}}

\addplot [color=mycolor4, line width=1.5pt]
  table[row sep=crcr]{%
-15.2	0.0216176377316355\\
-14	0.0284236338243975\\
-13.1	0.0348778497433599\\
-12.2	0.0427631078687316\\
-11.3	0.0523695838774021\\
-10.4	0.0640201838673757\\
-9.8	0.0730867033564007\\
-9.2	0.083295890411101\\
-8.6	0.0947068670387026\\
-8	0.107323124760658\\
-7.1	0.128254691322832\\
-5.6	0.16492886121371\\
-5	0.177531472079064\\
-4.7	0.182657071750809\\
-4.4	0.186699095497444\\
-4.1	0.189433989820881\\
-3.8	0.190659897717556\\
-3.5	0.190216934192831\\
-3.2	0.188007150325085\\
-2.9	0.184010702600089\\
-2.6	0.178294731146414\\
-2.3	0.171012456916943\\
-2	0.162391904230851\\
-1.4	0.142296856288977\\
-0.199999999999999	0.0991191109888607\\
0.399999999999999	0.0797310869795353\\
1	0.0630530912769025\\
1.6	0.0492474956831206\\
2.2	0.0381246073560071\\
2.8	0.0293302564288727\\
3.4	0.0224667545240536\\
4	0.0171577633977922\\
4.6	0.0130762154105888\\
5.2	0.00995140612985068\\
};
\addlegendentry{Interference as noise}

\end{axis}

\begin{axis}[%
width=5.833in,
height=4.375in,
at={(0in,0in)},
scale only axis,
xmin=0,
xmax=1,
ymin=0,
ymax=1,
axis line style={draw=none},
ticks=none,
axis x line*=bottom,
axis y line*=left
]
\end{axis}
\end{tikzpicture}%

%% file: fig3.tikz
\definecolor{mycolor4}{rgb}{0.49400,0.18400,0.55600}%
\begin{tikzpicture}[]
\hspace{-.3cm}
\begin{groupplot}[
 group style={
    group size=3 by 1,
    horizontal sep=40pt,
    group name=G},
    scale only axis,
unbounded coords=jump,
xmin=0,
xmax=1,
xlabel style={font=\color{white!15!black}},
xlabel={$R_1$},
ymin=0,
ymax=1,
ylabel style={font=\color{white!15!black}},
ylabel={$R_2$},
axis background/.style={fill=white},
axis x line*=bottom,
axis y line*=left,
xmajorgrids,
ymajorgrids,
ylabel near ticks,
xlabel near ticks,
legend style={font=\scriptsize,legend cell align=left,at={(1.6,1.2)},legend columns=3,column sep=10pt},
width=0.265\textwidth,
height=0.25\textwidth,
every axis/.append style={font=\footnotesize}]

\nextgroupplot[]
 \addplot [area legend, color=black,fill=green!20!white]
  table[row sep=crcr]{
  0 0\\
  0	0.39\\
  0.1	0.39\\
  0.39	0.1\\
  0.39	0\\
  };
%

\addplot[area legend, draw=black, fill=mycolor4]
table[row sep=crcr] {%
x	y\\
0	0\\
0.18	0\\
0.18	0.18\\
0	0.18\\
0	0\\
}--cycle;
%

\addplot [color=red, line width=2.0pt, only marks, mark=asterisk, mark options={solid, red}]
  table[row sep=crcr]{%
0.39	0\\
};
%

\addplot [color=red, line width=2.0pt, only marks, mark=asterisk, mark options={solid, red}, forget plot]
  table[row sep=crcr]{%
0	0.39\\
};
\addplot [color=red, dashed, line width=1.5pt, forget plot]
  table[row sep=crcr]{%
0.39	0\\
0.39	0.39\\
};
\addplot [color=red, dashed, line width=1.5pt, forget plot]
  table[row sep=crcr]{%
0	0.39\\
0.39	0.39\\
};
\node at (axis cs:0.9,0.9){\bf (a)};

\node (v1) at (0.1,0.39) {};
\node (v2) at (0.39,0.1) {};
\node (v3) at (0.39,0.39) {};
\draw[opacity=0.5,fill=red] (v1.center)--(v2.center)--(v3.center)--(v1.center);

\draw[green,line width=2.0pt] (0.39,0)  -- (0.39,0.1) ;
\draw[green,line width=2.0pt] (0,0.39)  -- (0.1,0.39) ;
\draw[blue,line width=2.0pt] (0.1,0.39)  -- (0.39,0.1) ;

\nextgroupplot[]
\addplot [area legend, color=black,fill=green!20!white]
  table[row sep=crcr]{
  0 0\\
  0	0.69\\
  0.14	0.69\\
  0.69	0.14\\
  0.69	0\\
  };
\addlegendentry{Theorem~\ref{thm:OBwdm}}

\addplot[area legend, draw=black, fill=mycolor4]
table[row sep=crcr] {%
x	y\\
0	0\\
0.09	0\\
0.09	0.09\\
0	0.09\\
0	0\\
}--cycle;
\addlegendentry{Interference as noise}

\addplot [color=red, line width=2.0pt, only marks, mark=asterisk, mark options={solid, red}]
  table[row sep=crcr]{%
0.59	0\\
};
\addlegendentry{Linear capacity -- single user~\cite{yousefi2015upper,kramer2015upper}}

\addplot [color=red, line width=2.0pt, only marks, mark=asterisk, mark options={solid, red}, forget plot]
  table[row sep=crcr]{%
0	0.59\\
};
\addplot [color=red, dashed, line width=1.5pt, forget plot]
  table[row sep=crcr]{%
0.59	0\\
0.59	0.59\\
};
\addplot [color=red, dashed, line width=1.5pt, forget plot]
  table[row sep=crcr]{%
0	0.59\\
0.59	0.59\\
};
\node at (axis cs:0.9,0.9){\bf (b)};

\node (v4) at (0.24,0.59) {};
\node (v5) at (0.59,0.24) {};
\node (v6) at (0.59,0.59) {};
\draw[opacity=0.5,fill=red] (v4.center)--(v5.center)--(v6.center)--(v4.center);

\node (d1) at (0.24,0.59) {};
\node (d2) at (0.14,0.69) {};
\node (d3) at (0,0.69) {};
\node (d4) at (0,0.59) {};
\node (d5) at (0.59,0.24) {};
\node (d6) at (0.69,0.14) {};
\node (d7) at (0.69,0) {};
\node (d8) at (0.59,0) {};
\draw[opacity=0.5,fill=gray] (d1.center)--(d2.center)--(d3.center)--(d4.center)--(d1.center);
\draw[opacity=0.5,fill=gray] (d5.center)--(d6.center)--(d7.center)--(d8.center)--(d5.center);

\draw[green,line width=2.0pt] (0.69,0)  -- (0.69,0.14) ;
\draw[green,line width=2.0pt] (0,0.69)  -- (0.14,0.69) ;
\draw[blue,line width=2.0pt] (0.14,0.69)  -- (0.69,0.14) ;

\nextgroupplot[every axis/.append style={font=\footnotesize}]
\addplot [area legend, color=black,fill=green!20!white]
  table[row sep=crcr]{
  0 0\\
  0	0.95\\
  0.16	0.95\\
  0.95	0.16\\
  0.95	0\\
  };

\addplot[area legend, draw=black, fill=mycolor4]
table[row sep=crcr] {%
x	y\\
0	0\\
0.04	0\\
0.04	0.04\\
0	0.04\\
0	0\\
}--cycle;

\addplot [color=red, line width=2.0pt, only marks, mark=asterisk, mark options={solid, red}]
  table[row sep=crcr]{%
0.83	0\\
};

\addplot [color=red, line width=2.0pt, only marks, mark=asterisk, mark options={solid, red}, forget plot]
  table[row sep=crcr]{%
0	0.83\\
};
\addplot [color=red, dashed, line width=1.5pt, forget plot]
  table[row sep=crcr]{%
0.83	0\\
0.83	0.83\\
};
\addplot [color=red, dashed, line width=1.5pt, forget plot]
  table[row sep=crcr]{%
0	0.83\\
0.83	0.83\\
};
\node at (axis cs:0.9,0.9){\bf (c)};

\node (v7) at (0.28,0.83) {};
\node (v8) at (0.83,0.28) {};
\node (v9) at (0.83,0.83) {};
\draw[opacity=0.5,fill=red] (v7.center)--(v8.center)--(v9.center)--(v7.center);

\node (e1) at (0.28,0.83) {};
\node (e2) at (0.16,0.95) {};
\node (e3) at (0,0.95) {};
\node (e4) at (0,0.83) {};
\node (e5) at (0.83,0.28) {};
\node (e6) at (0.95,0.16) {};
\node (e7) at (0.95,0) {};
\node (e8) at (0.83,0) {};
\draw[opacity=0.5,fill=gray] (e1.center)--(e2.center)--(e3.center)--(e4.center)--(e1.center);
\draw[opacity=0.5,fill=gray] (e5.center)--(e6.center)--(e7.center)--(e8.center)--(e5.center);

\draw[green,line width=2.0pt] (0.95,0)  -- (0.95,0.16) ;
\draw[green,line width=2.0pt] (0,0.95)  -- (0.16,0.95) ;
\draw[blue,line width=2.0pt] (0.16,0.95)  -- (0.95,0.16) ;
\end{groupplot}
\end{tikzpicture}